\documentclass[letterpaper, 10 pt, conference]{ieeeconf}  

\IEEEoverridecommandlockouts                              
\usepackage{mathptmx} 
\usepackage{times} 
\usepackage{amsmath} 
\usepackage{amssymb}  
\usepackage{graphicx}
\usepackage{tcolorbox}
\usepackage[utf8]{inputenc}
\usepackage{amsfonts}
\usepackage{dsfont}
\usepackage{bm}
\usepackage[backend=bibtex, style=ieee, doi=false, url=false, isbn=false, date=year]{biblatex}
\usepackage{xurl}
\usepackage{hyperref}
\usepackage{todonotes}
\usepackage{comment}
\usepackage{circuitikz}

\graphicspath{{./Figures/}}

\DeclareMathOperator{\sign}{sign}

\newcommand{\R}{\mathbb{R}}

\newtheorem{theorem}{Theorem}  

\newtheorem{proposition}{Proposition}

\newtheorem{remark}{Remark}

\newtheorem{example}{Example}

\title{\LARGE \bf
Singular Port-Hamiltonian Systems Beyond Passivity
}

\author{Henrik Sandberg, Kamil Hassan, Heng Wu 
\thanks{This work was supported in part by the Swedish Research Council (Grant~2023-04770), Nordic Energy Research (project PBC-IBR), and the Swedish Energy Agency through the European CETPartnership (project ProRES).}
\thanks{H. Sandberg and K.~Hassan are with the Department of Decision and Control Systems, KTH Royal Institute of Technology, Stockholm, Sweden
        {\tt\small \{hsan,kamilha\}@kth.se}}%
\thanks{H.~Wu, formerly with Aalborg University, Denmark, is now with the School of Electrical Engineering, 
Southeast University, Nanjing, China {\tt\small hengwu@seu.edu.cn}}
}%

\begin{document}
\maketitle
\thispagestyle{empty}
\pagestyle{empty}

\begin{abstract}
In this paper, we investigate a class of port-Hamiltonian systems with singular vector fields. We show that, under suitable conditions, their interconnection with passive systems ensures convergence to a prescribed non-equilibrium steady state. At first glance, this behavior appears to contradict the seemingly passive structure of port-Hamiltonian systems, since sustaining a non-equilibrium steady state requires continuous power injection. We resolve this apparent paradox by showing that the singularity in the vector field induces a sliding mode that contributes effective energy, enabling maintenance of the steady state and demonstrating that the system is not passive. Furthermore, we consider regularizations of the singular dynamics and show that the resulting systems are cyclo-passive, while still capable of supplying the required steady-state power. These results clarify the role of singularities in port-Hamiltonian systems and provide new insight into their energetic properties.
\end{abstract}

\section{Introduction}\label{sec:intro}
In the recent paper~\cite{kong_control_2024}, a dynamical model with an apparent port-Hamiltonian structure was proposed for power inverter control. Remarkably, despite its seemingly passive external behavior, the system operates as a power source in steady state. This apparent paradox challenges the standard intuition associated with passive port-Hamiltonian systems. 

A closer inspection reveals that the underlying mechanism stems from singularities in the vector field, which modify the system’s energy-based interpretation. Motivated by this observation, this paper introduces a simple yet fairly general class of port-Hamiltonian systems with singular vector fields and provides a theoretical mechanism to explain this unconventional behavior.

In~\cite{kong_control_2024}, the authors adopt a passivity-based framework to argue for the stability of the inverter controller when interconnected with a passive system (i.e., a load). However, since the inverter operates as a power source in steady state, standard passivity-based arguments are not directly applicable to guarantee convergence to the desired operating point. This limitation is closely related to the so-called \emph{dissipation obstacle} discussed in~\cite{ortega_putting_2001}, which highlights that when stabilization of a desired steady state requires nonzero dissipation, a purely passive controller is insufficient.

Dissipativity theory, together with passivity- and energy-based control methods, has a long and successful history in the control literature; see, e.g.,~\cite{willems_dissipative_1972, moylan_dissipative_2014,ortega_putting_2001} and the references therein. In this work, we primarily adopt a port-Hamiltonian framework~\cite{van_der_schaft_port-hamiltonian_2014}. The class of singular port-Hamiltonian systems considered here bears similarities to the notion of cyclo-dissipativity~\cite{willems_qualitative_1974,van_der_schaft_cyclo-dissipativity_2021} (or cyclo-passivity~\cite{moylan_dissipative_2014}), a connection that will be further discussed in the sequel.

The contributions of this paper are summarized as follows:
\begin{enumerate}
    \item Singular port-Hamiltonian systems: We introduce a class of singular port-Hamiltonian systems and show that their passivity and possibility to act as an infinite energy source depend critically on how the output is defined on a singular set of measure zero.
    \item Stability under interconnection: Using passivity and sliding-mode arguments, we derive conditions under which the singular systems, when interconnected with static passive systems, converge in finite time to a prescribed singular set and can act as a continuous energy source on that set.
    \item Implementable approximations: We propose two implementable approximations of the singular systems, in which the infinite singularity is regularized, and show that one satisfies a modified dissipation inequality outside the singular set, while the other is cyclo-passive.
\end{enumerate}

The remainder of this paper is organized as follows. Section~\ref{sec:pH} reviews the port-Hamiltonian framework and introduces the proposed systems featuring a singular set. Section~\ref{sec:stab_interconnect} presents a stability analysis of the open-loop systems and of their closed-loop interconnection with static passive systems.
Section~\ref{sec:singularity} develops two implementable system approximations. Finally, Sections~\ref{sec:examples} and~\ref{sec:conclusion} provide numerical simulations and concluding remarks, respectively, including a discussion of open research directions.

\subsubsection*{Notation} For a scalar-valued function~$H:\R^n \to \R$, we define the gradient as the column vector~$\nabla H(x) := \begin{bmatrix}
    \frac{\partial H}{\partial x_1} & \frac{\partial H}{\partial x_2} & \dots & \frac{\partial H}{\partial x_n}
\end{bmatrix}^\top \in \R^{n\times 1}$. We use $\dot{H}$ to denote the total time derivative of $H$. That is, $\dot H(x):= \frac{d}{dt}H(x(t)) = \nabla H(x(t))^\top \dot x(t)$ (leaving out the time argument~$t$ in the following when it is clear from the context). Vectors $x,\,\dot x\in\mathbb{R}^n$ are identified with column vectors of dimension $n\times 1$. For symmetric matrices $A$ and $B$, $A \prec B$ ($A \preceq B$) means $A-B$ is negative definite (semidefinite), and $\lambda_{\max}(A)$ is the largest eigenvalue of~$A$.

\section{Port-Hamiltonian Systems}  \label{sec:pH}
In this section, we first review the formalism of port-Hamiltonian systems~\cite{van_der_schaft_port-hamiltonian_2014} and then introduce the class of singular systems under consideration.

\subsection{Port-Hamiltonian Systems}
Input--state--output port-Hamiltonian systems~\cite{van_der_schaft_port-hamiltonian_2014} are dynamical state-space systems in the form
\begin{equation} \label{eq:pH}
    \begin{aligned}
        \dot x & = [J(x,t) - R(x,t)]\nabla H(x) + B(x)u \\
        y & = B(x)^\top \nabla H(x),
    \end{aligned}
\end{equation}
where $x\in\mathbb{R}^n$ (state), $u,\,y\in\mathbb{R}^m$ (input, output), $H(x)\in\mathbb{R}$ (Hamiltonian/storage function), $ J(x,t)^\top=- J(x,t)\in\mathbb{R}^{n\times n}$ (interconnection matrix), $R(x,t)=R(x,t)^\top \in\mathbb{R}^{n\times n}$ (dissipation matrix), and $B(x)\in\mathbb{R}^{n\times m}$ (input matrix), for all $x,t$. For the remainder of the paper, \emph{we assume~\eqref{eq:pH} admit solutions~$x(t)$, $t\geq 0$, in the sense of Filippov~\cite{cortes_discontinuous_2008}, for any initial state $x(0)$ and applied control $u=u(t)$, $t\geq 0$.}
We also make the standing assumption that the dissipation matrix~$R(x,t)$ is uniformly positive definite: There are positive constants $\epsilon_1,\,\epsilon_2$ such that $\epsilon_1 I \prec R(x,t)\prec \epsilon_2 I$, for all $x,\, t$.
The system~\eqref{eq:pH} then satisfies the differential \emph{dissipation inequality}~\cite{willems_dissipative_1972}
\begin{equation}\label{eq:diss_ineq}
\begin{aligned}
    \dot{H}(x) & = \nabla H(x)^\top \dot{x} \\ & = -\nabla H(x)^\top R(x,t) \nabla H(x) + y^\top u  \leq y^\top u,
\end{aligned}
\end{equation}
by observing $\nabla H(x)^\top J(x,t) \nabla H(x) \equiv 0$.
The Hamiltonian~$H(x)$ is thus a storage (``energy'') function, $y^\top u$ a supply rate (``power injection''), and the system~\eqref{eq:pH} is \emph{passive} if the \emph{storage functions are bounded from below~\cite{willems_dissipative_1972}}. Since we can always add a constant to a Hamiltonian without changing the dynamics, for passive systems, it is henceforth assumed that the minimum value of $H(x)$ is zero. If there is no lower bound on $H(x)$, the system~\eqref{eq:pH} is \emph{cyclo-dissipative with supply rate $y^\top u$}~\cite{willems_qualitative_1974,van_der_schaft_cyclo-dissipativity_2021}, or simply \emph{cyclo-passive}~\cite{moylan_dissipative_2014}. The name derives from the fact that any closed state trajectory in~$\mathbb{R}^n$ satisfies
\begin{equation}\label{eq:cyclodiss}
    \oint y^\top u \,\,dt \geq 0. 
\end{equation}
That is, we cannot retrieve a net supply from the system over any closed cycle. Examples of cyclo-dissipative systems include circuits with positive resistors and possibly negative inductors and capacitors. 

\subsection{A Class of Singular Port-Hamiltonian Systems}
This paper is concerned with the energetic properties of a class of \emph{singular} port-Hamiltonian systems. Their Hamiltonians and input matrices are given by
\begin{equation} \label{eq:pHalmost}
    H(x) = \frac{1}{2} \sigma(x)^2,  \qquad
         B(x)  = \sigma^{-1}(x)\bar B(x)
\end{equation}
with the quadratic form
\begin{equation}\label{eq:sigma}
  \sigma(x) := \frac{1}{2}(x^\top Q x - 1),
\end{equation}
and $\bar B(x)$ continuous and nonzero for all~$x$. 
We assume $Q=Q^\top \in\mathbb{R}^{n\times n}$ is positive definite and introduce the set
\begin{equation*}
    \mathcal{S}=\{x \in \mathbb{R}^n:\,\sigma(x)=0 \Leftrightarrow x^\top Q x = 1\}.
\end{equation*}
The set $\mathcal{S}$ defines the boundary of an ellipsoid and is called the \emph{singular set} of states, since $B(x)$ is singular there. The Hamiltonian~$H$ is a quartic polynomial and 
the set~$\mathcal{S}$ consists of all $x$ of minimum storage: $H(x)=0$.
We have $\sigma(x)<0$ inside of and $\sigma(x)>0$ outside of~$\mathcal{S}$. Since $\bar B(x)$ is nonzero and continuous, the sign of the input matrix~$B(x)$ flips when crossing the ellipsoidal surface~$\mathcal{S}$, and the gain approaches infinity in its vicinity. 

With the choices made in~\eqref{eq:pHalmost}--\eqref{eq:sigma}, the port-Hamiltonian system~\eqref{eq:pH} attains the form:
\begin{equation} \label{eq:sys_nom}  \tag{SpH'}
    \begin{aligned}
        \dot x & = [J(x,t) - R(x,t)]\sigma(x)Qx + \sigma^{-1}(x)\bar B(x)u \\
        y & = c(x)\bar B(x)^\top Qx,
    \end{aligned}
\end{equation}
since
\begin{equation*}
    \nabla H(x)=\sigma(x)Qx,
\end{equation*}
and introducing $c(x):=\sigma(x)/\sigma(x)$. 
Clearly, $c(x)=1$ for $x\not\in\mathcal{S}$. To clarify the value on the set~$\mathcal{S}$, consider the following passivity argument:
The time derivative of the Hamiltonian along trajectories of~\eqref{eq:sys_nom} is
\begin{equation}\label{eq:diss_ineq_nom}
    \dot H(x) = \sigma(x)x^\top Q \dot x = -\sigma^2(x) x^\top Q R(x,t) Q x + y^\top u,
\end{equation}
noting $x^\top Q J(x,t) Qx \equiv 0$. Since the storage $H(x)\geq 0$ and the dissipation rate $\sigma^2(x) x^\top Q R(x,t) Q x\geq 0$, for all $x,\,t$, \eqref{eq:sys_nom} should be passive. This requires $c(x)=0$ for $x\in\mathcal{S}$ and
\begin{equation}\label{eq:c_x}
    c(x) = \mathbf{1}_{\mathcal{S}^\text{c}}(x),
\end{equation}
where $\mathbf{1}_{\mathcal{S}^\text{c}}$ is the indicator function on the complement of~$\mathcal{S}$. Equation~\eqref{eq:c_x} ensures a zero output~$y$ for~$x\in\mathcal{S}$ and precludes the possibility of extracting supply from states where $H(x)=0$. 
Example~\ref{ex:circuitexample} elaborates further on this point in a concrete scenario.
When $R(x,t)\equiv 0$, the system is $\emph{lossless}$~\cite{willems_dissipative_1972}. 
\begin{remark}
    We consider solutions to~\eqref{eq:sys_nom} in the sense of Filippov and only specify $\sigma^{-1}(x)\in (-\infty,\infty)$, $x \in \mathcal{S}$; see~\cite{utkin_sliding_1992,slotine_applied_1991,cortes_discontinuous_2008}. Bounding the singularity set is treated in Section~\ref{sec:singularity}.
\end{remark}

\begin{example}[\eqref{eq:sys_nom} connected to resistor~$R_\text{test}$] \label{ex:circuitexample}
\begin{figure}[tb]
    \centering
    \begin{circuitikz}[american voltages,scale=0.85,transform shape]
    \draw
    (0,0) to[V, l_=SpH'/SpH, invert, v^>={$v=y$}] (0,3)
          to[short, i>^={$i=-u$}] (4,3)
          to[R, l_=$R_\text{test}$] (4,0)
          -- (0,0);
    \end{circuitikz}
    \caption{Singular port-Hamiltonian system modeling a voltage source, connected to a test resistor~$R_\text{test}>0$.}
    \label{fig:circuit}
\end{figure}
Consider the circuit in Fig.~\ref{fig:circuit} and the scalar lossless voltage source dynamics with internal state $x$,
\begin{equation}\label{eq:SpH_ex}
    \dot x = \sigma^{-1}(x)u = \frac{2u}{x^2-1}, \quad y = c(x)x.
\end{equation}
Note this dynamics satisfies~\eqref{eq:sys_nom} with $J=R=0$, $Q=\bar B=1$, $\sigma(x)=(x^2-1)/2$, $H(x)=(x^2-1)^2/8$, and thus $\mathcal{S}=\{-1,1\}$. Let the output denote the voltage across the circuit, $v=y$, and let the input denote the current injected into the source, $i=-u$. Given the connected test resistance, we have $u=-y/R_\text{test}$ and the closed circuit dynamics
\begin{equation} \label{eq:circuit_dyn}
    \dot x = -\frac{2c(x)x}{R_\text{test}(x^2-1)}, \quad y=c(x)x.
\end{equation}
It is easy to verify that for any initial voltage $x(0)=y(0) \in (0,\infty)\setminus \{1\}$, we have $x(t)\to 1$ with \emph{finite-time convergence}\footnote{The solution $x(t)\equiv 1$ is a stable sliding mode since the flow pushes into it from both sides; see~\cite{utkin_sliding_1992,slotine_applied_1991}.}. The dissipated energy into the resistor is $H(x(0))>0$ during the finite-time transient of duration~$t_\text{conv}$, since the source is lossless. 
After convergence to~$1\in\mathcal{S}$, the applied voltage is $y(t)=c(1)=0$, $t\geq t_\text{conv}$, per~\eqref{eq:c_x}. Hence, the voltage instantaneously switches to zero once it reaches~$v=1$, and the source stops delivering power to the resistor. This is necessary because of the passivity of~\eqref{eq:sys_nom}: The storage is at a minimum at $x=1$, and its internal energy is depleted.

At this point, it can be of interest to contemplate the implications of redefining $c(x)\equiv 1$ in~\eqref{eq:SpH_ex}--\eqref{eq:circuit_dyn}. This change seems rather innocent, but has consequences: The trajectories~$x(t)$ are identical, but after convergence it holds $y(t)=c(1)=1$, $t\geq t_\text{conv}$. There, $H(1)=0$, and yet the source delivers an infinite amount of energy over $[t_\text{conv},\infty)$, since a constant nonzero voltage is applied to the resistor indefinitely. Hence, the voltage source~\eqref{eq:SpH_ex} is \emph{no longer a passive system}, even if we just redefined $c(x)$ on the set~$\mathcal{S}$ of zero measure. 
From a practical perspective, the choice $c(x)\equiv 1$ is of interest if the goal is to build a voltage source that delivers a fixed, continuous voltage to passive loads with guaranteed finite-time convergence. \hfill $\blacksquare$
\end{example}

Example~\ref{ex:circuitexample} illustrates that if we redefine $c(x)\equiv 1$, \eqref{eq:sys_nom} may look passive everywhere, except on a set of measure zero. 
This redefinition can be practically useful when designing control systems that behave passively ``almost everywhere.'' Such systems can ensure convergence to a desired target set when interconnected with passive loads, while still allowing for steady outputs that would otherwise be incompatible with strict passivity, as they require an infinite energy supply. In the rest of the paper, we shall therefore study the following systems: 
\begin{tcolorbox}[colback=white,colframe=black,boxrule=0.5pt,top=2pt,
    bottom=2pt]
\begin{equation} \label{eq:SpH}  \tag{SpH}
    \begin{aligned}
        \dot x & = [\bar J(x,t) - \sigma(x)R(x,t)]Qx + \sigma^{-1}(x)\bar B(x)u \\
        y & = \bar B(x)^\top Qx,
    \end{aligned}
\end{equation}
\end{tcolorbox}
\noindent
Here we have introduced $J(x,t)=:\sigma^{-1}(x)\bar J(x,t)$, assuming $\bar J(x,t)$ is continuous everywhere, and again used the convention  $\sigma(x)/\sigma(x)\equiv 1$, to allow for nonzero, continuous autonomous flow across the singular set. Since $J$ and $\bar J$ are anti-symmetric, they do not contribute to the energy exchange, and this choice is less delicate than the choice of output. 
We will generalize Example~\ref{ex:circuitexample} for~\eqref{eq:SpH} in Section~\ref{sec:stab_interconnect}. 
In Section~\ref{sec:singularity}, we regularize~$\sigma^{-1}(x)$ to avoid infinite singularities in the model, and show that 
several properties of~\eqref{eq:SpH} are retained. Numerical examples in Section~\ref{sec:examples} will serve to illustrate that complex convergence behavior can occur when~\eqref{eq:SpH} is interconnected to dynamic passive systems. In particular, after~\eqref{eq:SpH} injects energy into a connected load, it is uncertain whether the system consistently reaches the desired steady state.

\begin{remark}
 The port-Hamiltonian inverter controllers proposed in~\cite{kong_control_2024} were introduced to ensure convergence to a desired periodic steady state that requires nonzero power to sustain when interconnected to passive loads. No proofs were provided there, but Example~\ref{ex:circuitexample} and the analysis to follow will illustrate a mechanism for similar behavior.
\end{remark}

\begin{remark}\label{rem:IDA-PBC}
Using standard port-Hamiltonian systems, the general Interconnection and Damping Assignment Passivity-Based Control (IDA--PBC) framework~\cite{ortega_putting_2001, ortega_interconnection_2002} enables the design of controllers that are “almost passive,” guaranteeing convergence to steady states that would otherwise require an infinite energy supply. In IDA--PBC, this is achieved by introducing an infinite-energy source and employing a state-modulated interconnection structure. In contrast, in~\eqref{eq:SpH}, the energy supply is implicitly encoded in the set~$\mathcal{S}$. While the systems in~\eqref{eq:SpH} are restrictive, they do not involve solutions of partial differential equations as in IDA--PBC. A detailed comparison is beyond the scope of this paper and is left for future work.
\end{remark}

\section{Stability and Feedback Interconnection} \label{sec:stab_interconnect}
The following result characterizes the equilibrium and steady state of the uncontrolled system~\eqref{eq:SpH}.
\begin{proposition}\label{prop:SpH}
The uncontrolled system~\eqref{eq:SpH} ($u\equiv 0$) satisfies the following properties:
\begin{enumerate}
    \item[(i)] The equilibrium $x(t)\equiv 0$ is repelling.
    \item[(ii)] The set~$\mathcal{S}$ is forward-invariant and the trajectories satisfy $\dot x = \bar J(x,t)Qx$ on $\mathcal{S}$.
    \item[(iii)] All trajectories $x(t)$, $t\geq 0$, with inital condition  $x(0)\neq 0$, converge to the singular set $\mathcal{S}$.
    \end{enumerate}
\end{proposition}
\begin{proof}
    (i): Clearly, $x=0$ is an equilibrium point. To prove instability, we follow the argument from the instability theorem in~\cite[Theorem~4.3]{khalil_nonlinear_2002}: Consider the positive definite function $V(x)=\frac{1}{2}x^\top Q x$. We have that $\dot V(x) = x^\top Q \dot x = -\sigma(x) x^\top QRQ x$, since $x^\top Q\bar JQx \equiv 0$. Furthermore, $\sigma(x)$ is uniformly \emph{negative} in a neighborhood of $x=0$, and $\dot V(x)>0$ for all $x\neq 0$ in the same neighborhood. Hence, for any $x(0)$ arbitrarily close to (but not equal to) zero, $V(x(t))$ must be an \emph{increasing} function, which proves the equilibrium point is repelling (it has no stable manifold).
    (ii): We have $x\in\mathcal{S} \Leftrightarrow x^\top Q x = 1$. The normal direction of the ellipsoidal surface~$\mathcal{S}$ is $\nabla (x^\top Q x - 1) = 2Qx$. Since the vector field of~\eqref{eq:SpH} evaluated on~$\mathcal{S}$ ($\dot x = \bar J(x,t)Qx$) is perpendicular to the normal vector at every point, $(2Qx)^\top \dot x \equiv 0$,  we conclude the set $\mathcal{S}$ is invariant (by Nagumo’s theorem).
    (iii): We apply LaSalle's theorem~\cite[Theorem~4.4]{khalil_nonlinear_2002}: Consider the non-negative Hamiltonian $H(x)=\frac{1}{2}\sigma(x)^2$, whose level sets $\{H(x)=c\}$ are compact and non-empty for all $c>0$. The function satisfies
    \begin{equation*}
        \dot H(x) = \sigma(x) \dot \sigma(x) = \sigma(x) x^\top Q \dot x = -\sigma(x)^2 x^\top QRQx \leq 0.
    \end{equation*}
    The set of points where $\dot H(x)=0$ is $E:=\mathcal{S} \cup \{0\}$. The set $E$ is forward invariant by~(i)--(ii). Hence all trajectories $x(t)$ will converge to $E$ as $t\to \infty$. However, the point $x=0$ is repelling~(i), so any $x(0)\neq 0$ will have to converge to $\mathcal{S}$.
\end{proof}

Consider~\eqref{eq:SpH} under negative feedback interconnection with a possibly nonlinear and continuous static passive system:
\begin{equation}\label{eq:static_load}
\begin{aligned}
    u  & = -K(y) y, \\
    K(y)& =K(y)^\top \in\mathbb{R}^{m \times m},\quad
    \kappa_1 I  \prec K(y) \prec \kappa_2 I,
\end{aligned}
\end{equation}
for some positive constants $\kappa_1,\,\kappa_2$, and all $y$. 
The closed-loop dynamics become
\begin{equation}\label{eq:feedback}
\begin{aligned}
     \dot x  & = [\bar J(x,t) - \sigma(x)R(x,t)]Qx - \sigma^{-1}(x)\bar B(x)K(y)\bar B(x)^\top Qx \\
     & = [\bar J(x,t) - \sigma(x)R(x,t) - \sigma^{-1}(x)\bar B(x)K(y)\bar B(x)^\top] Qx.
\end{aligned}
\end{equation}
Comparing~\eqref{eq:feedback} with the open-loop system considered in Proposition~\ref{prop:SpH}, we note that there is a new term involving~$\sigma^{-1}(x)$ in the vector field. While the solution $x(t)\equiv 0$ remains an (unstable) equilibrium (the proof is essentially the same as in Proposition~\ref{prop:SpH}), the set~$\mathcal{S}$ is no longer a classical solution but rather a potential sliding mode/surface~\cite{utkin_sliding_1992,slotine_applied_1991} due to the discontinuity.
Hence, we consider $\sigma(x)=0$ as a potential sliding mode and follow a standard argument~\cite{slotine_applied_1991} to check whether the dynamics push trajectories towards it. Generally, this is ensured by verifying $\frac{d}{dt}\frac{1}{2}\sigma(x(t))^2 = \sigma(x)\dot \sigma(x)<0$. Incidentally, this condition coincides with $\dot H$: 
\begin{align*}
     \dot H(x) & = \sigma(x)\dot \sigma(x) = \nabla H(x)^\top \dot x   \\
     & =  -x^\top Q (\sigma^2(x)R(x,t) +\bar B(x)K(y)\bar B(x)^\top ) Q x  < 0,
\end{align*}
for $x\neq 0$. Hence, the trajectories starting in $x(0)\neq 0$ indeed approach $\mathcal{S}$ as summarized next.
\begin{proposition}
    Consider the interconnection of~\eqref{eq:SpH}  and~\eqref{eq:static_load}, as modeled in~\eqref{eq:feedback}. Then every trajectory~$x(t)$ with $x(0)\neq 0$ converges to the set~$\mathcal{S}$.
\end{proposition}

In simulations, we observe that trajectories tend to reach the set~$\mathcal{S}$ in \emph{finite time} and stay there (see Section~\ref{sec:examples} and Example~\ref{ex:circuitexample}). This can be proven under certain conditions, as formalized next.
\begin{proposition}\label{prop:impact_time}
    Consider the interconnection of~\eqref{eq:SpH} and~\eqref{eq:static_load}, as modeled in~\eqref{eq:feedback}, and assume $\bar B\in\mathbb{R}^{n\times n}$ is invertible. Also consider the set
\begin{equation*}
\mathcal{S}_l^\text{c}:=\{x:\,\|x\|_2 \geq l,\, x\not \in \mathcal{S}\},
\end{equation*}
for some $l$ such that $0<l<1/\lambda_{\max}^{1/2}(Q)$ (such that $\|x\|_2=l$ lies completely inside the ellipsoid $\mathcal{S}$).
    \begin{enumerate}
        \item[(i)] Every trajectory $x(t)$ with $x(0)\in\mathcal{S}_l^c$ reaches the set~$\mathcal{S}$ in finite time~$t_\text{conv}$, which satisfies $t_\text{conv}\leq H(x(0))/\eta$, where $\eta$ is defined in~\eqref{eq:diss_ineq_fb}.
        \item[(ii)] The set $\mathcal{S}$ is a sliding manifold, and the corresponding sliding-mode dynamics are given by $\dot x = \bar J(x,t)Qx$.
    \end{enumerate}
\end{proposition}
\begin{proof}
(i): We can bound the dissipation rate in~$\mathcal{S}_l^c$ as
\begin{equation} \label{eq:diss_ineq_fb}
\begin{aligned}
     \dot H(x) & = -x^\top Q (\sigma^2(x)R(x,t) +\bar BK(y)\bar B^\top ) Q x \\
     & \leq -\kappa_1 x^\top Q \bar B\bar B^\top  Q x \leq -\kappa_1 l^2 \lambda_{\min}(Q \bar B\bar B^\top  Q) \\ & =: -\eta <0.
\end{aligned}
\end{equation}
Since $H$ decreases at least at a rate~$\eta$, the time to reach $\mathcal{S}$ where $H(x)=0$ is upper bounded by $H(x(0))/\eta$. 
(ii): On the manifold~$\mathcal{S}$ ($\sigma(x)=0$), sliding-mode solutions should satisfy the differential inclusion $\dot x \in \text{co}\{f^+(x,t),f^-(x,t)\}$~\cite[Section~7.1.2]{slotine_applied_1991}. Here $f^+$ is the vector field just outside $\mathcal{S}$ ($\sigma(x)=\epsilon$, for arbitrarily small $\epsilon>0$) and $f^-$ is the field just on the inside ($\sigma(x)=-\epsilon$). In this case $f^+(x,t)=[\bar J(x,t) - \frac{1}{\epsilon}\bar BK(y)\bar B)^\top] Qx$ and $f^-(x,t)=[\bar J(x,t) + \frac{1}{\epsilon}\bar BK(y)\bar B^\top] Qx$. Note that $\sigma^{-1}$ flips sign on the manifold. We see that $\nabla \sigma(x)^\top f_+(x,t)<0$ and $\nabla \sigma(x)^\top f_-(x,t)>0$ such that the fields push into every point on $\mathcal{S}$ from both sides. The convex combination $\frac{1}{2}f_+ + \frac{1}{2}f_- =\bar J(x,t)Qx$ is tangential to~$\mathcal{S}$ everywhere, proving that a sliding-mode solution satisfies this dynamics.
\end{proof}

To summarize, once the closed-loop system~\eqref{eq:feedback} reaches the ellipsoid~$\mathcal {S}$, there is a sliding-mode solution satisfying
\begin{align*}
     \dot x  & = \bar J(x,t) Qx \\
    y & = \bar B(x)^\top Qx, \quad u = -K(y)y.
\end{align*}
This is a generalization of Example~\ref{ex:circuitexample} for~\eqref{eq:SpH}. In the two-dimensional case, 
if the system designer chooses $\bar J(x,t) = \left[ \begin{smallmatrix}
    0 & -\omega_0 \\ \omega_0 & 0
\end{smallmatrix}\right]$, then $x$ and $y$  in steady state will oscillate with frequency~$\omega_0$ of a desired amplitude (controlled by $Q$), and dissipate unbounded amount of energy into $K(y)$. 

As with all sliding mode controllers, switching can lead to chattering in practice. The problem is here exacerbated by $\sigma^{-1}(x)$ approaching $\pm\infty$ at the switching manifold. Next, it is shown that some of the desirable properties of~\eqref{eq:SpH} are retained (in an appropriate sense) when $\sigma^{-1}(x)$ is regularized.

\section{Managing the Singularity}\label{sec:singularity}
\begin{figure}[tb]
    \centering
    \includegraphics[width=0.65\linewidth]{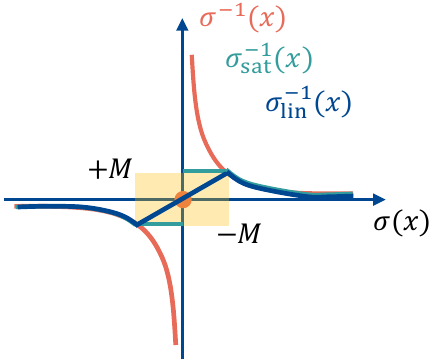}
    \caption{Illustration of the two bounded approximations of $\sigma^{-1}(x)$. The function~$\sigma^{-1}(x)$ is clipped to the interval~$[-M,M]$. $\sigma_\text{sat}^{-1}(x)$ retains a finite discontinuity at $\sigma(x)=0$, whereas $\sigma_\text{sat}^{-1}(x)$ introduces a linear interpolation in the clipped region.}
    \label{fig:sigma_inverses}
\end{figure}
Here, we investigate whether it is possible to retain essential properties of~\eqref{eq:SpH}\footnote{The possibility to act as an infinite energy source on~$\mathcal{S}$ and to satisfy a passivity inequality outside of~$\mathcal{S}$.} while avoiding the infinite singularity on~$\mathcal{S}$.
Consider the two approximations of~$\sigma^{-1}(x)$ illustrated in Fig.~\ref{fig:sigma_inverses}. The first one ($\sigma^{-1}_\text{sat}$) saturates the infinite discontinuity, introducing a finite jump discontinuity~$\pm M$. The second one ($\sigma^{-1}_\text{lin}$) approximates the discontinuity using a steep, linear interpolation in a boundary region (compare with~\cite[Section~7.2]{slotine_applied_1991}). It should be clear that as $M\to\infty$, both approximations converge to $\sigma^{-1}(x)$.

\subsection{Finite Jump Discontinuous Approximation}
Consider the set
\begin{equation*}
    \mathcal{M}:=\{x\in\mathbb{R}^n:\,|\sigma(x)|\leq 1/M \},
\end{equation*}
for some fixed parameter $M>0$. Note that $\mathcal{S} \subset \mathcal{M}$ and that $\mathcal{M}$ converges to $\mathcal{S}$ as $M\to\infty$; see Fig.~\ref{fig:Sets}.
\begin{figure}[tb]
    \centering
    \includegraphics[width=0.8\columnwidth]{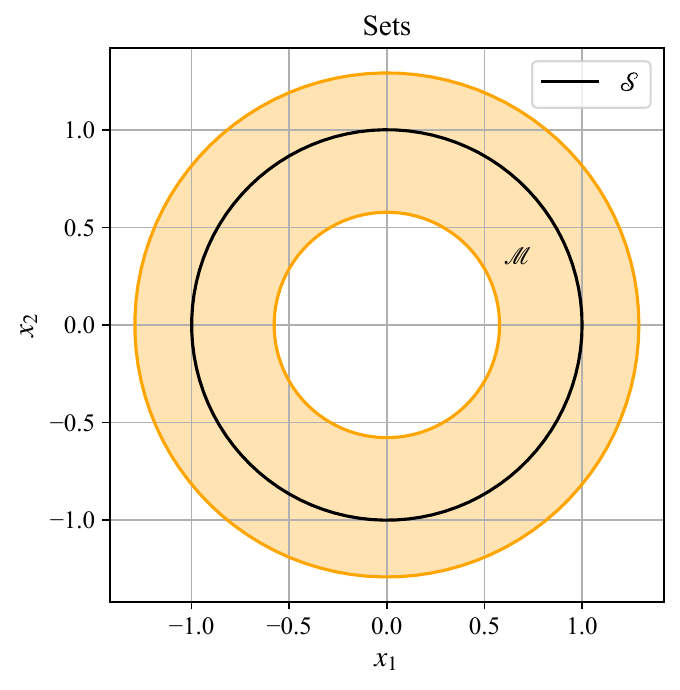}
    \caption{The sets $\mathcal{S}$ (black) and $\mathcal{M}$ (orange), for $Q=I$, $M=3$, and $n=2$.}
    \label{fig:Sets}
\end{figure}
Let us now saturate $\sigma^{-1}$ inside $\mathcal{M}$,
\begin{equation*}
    \sigma^{-1}_\text{sat}(x) := 
    \begin{cases}
         1/\sigma(x), & \text{if $x \in \mathcal{M}^\text{c}:=\mathbb{R}^n\setminus \mathcal{M}$} \\
         M \sign \sigma(x), & \text{if $x\in \mathcal{M}$}
    \end{cases}.
\end{equation*}
The ``inverse'' $\sigma^{-1}_\text{sat}$ has a \emph{finite} jump discontinuity from $-M$ to $M$ when crossing the ellipsoid~$\mathcal{S}$. Consider the saturated (``regularized'') version of~\eqref{eq:SpH}:
\begin{equation} \label{eq:sys_sat}
    \begin{aligned}
        \dot x & = [\bar J(x,t) - \sigma(x)R(x,t)]Qx + \sigma^{-1}_\text{sat}(x)\bar B(x)u   \\
        y & = \bar B(x)^\top Qx.
    \end{aligned}
\end{equation}
In~$\mathcal{M}^\text{c}$, $\eqref{eq:sys_sat}$ is identical to~\eqref{eq:SpH}. In the saturated region~$\mathcal{M}$, we obtain by direct differentiation of $H$ the equality
\begin{equation*}
    \dot H(x) = \sigma(x)x^\top Q \dot x = -\sigma^2(x) x^\top QR(x,t)Q x + M|\sigma(x)|y^\top u.
\end{equation*}
Due to the factor $M|\sigma(x)|$, $H$ is not a storage function with respect to the supply rate  $y^\top u$. However, upon division by~$M|\sigma(x)|$, we obtain the differential dissipation equality
\begin{equation*}
     \frac{d}{dt}\frac{|\sigma(x)|}{M} = \frac{1}{M}\sign \sigma(x) \dot{\sigma}(x) =\underbrace{-\frac{|\sigma(x)|}{M} x^\top QRQ x}_{\leq 0} + y^\top u,
\end{equation*}
for $x \in \mathcal{M}\setminus \mathcal{S}$.
Hence, for $x\not\in\mathcal{S}$, the system~\eqref{eq:sys_sat} satisfies a dissipation inequality with supply rate $y^\top u$ and storage function
\begin{equation}\label{eq:Hsat}
    H_\text{sat}(x) = 
    \begin{cases}
        H(x) + \frac{1}{2M^2}, &  \text{if $x\in \mathcal{M}^\text{c}$} \\
        \frac{|\sigma (x)|}{M}, & \text{if $x\in  \mathcal{M}$}
    \end{cases}.
\end{equation}
The constant term $\frac{1}{2M^2}$ has been introduced to achieve continuity of the storage function across the boundary~$\partial \mathcal{M}$ of $\mathcal{M}$. This constant, of course, does not affect the dynamics. The gradient~$\nabla H_\text{sat}(x)$ is continuous, also across the boundary~$\partial \mathcal{M}$.
We summarize the result in the following theorem.
\begin{theorem}\label{thm:diss_sat}
Consider the regularized system~\eqref{eq:sys_sat}. Trajectories~$x(t) \not \in \mathcal{S}$ satisfy the differential dissipation inequality
\begin{equation*}
    \dot{H}_\text{sat}(x) = -d_\text{sat}(x,t) + y^\top u \leq y^\top u,
\end{equation*}
with storage function $H_\text{sat}\geq 0$ in~\eqref{eq:Hsat} and
continuous dissipation rate
\begin{equation}
    d_\text{sat}(x,t) = 
    \begin{cases}
        \sigma^2(x) x^\top Q R(x,t) Q x, & \text{if $x\in \mathcal{M}^\text{c}$} \\
         \frac{|\sigma(x)|}{M} x^\top QR(x,t)Q x, & \text{if $x\in  \mathcal{M}$}
    \end{cases}.
\end{equation}
\end{theorem}

The system~\eqref{eq:sys_sat} has a milder discontinuity compared to~\eqref{eq:SpH}, but still satisfies a passivity-inequality outside~$\mathcal{S}$ \emph{and} admits a sliding-mode solution on~$\mathcal{S}$ in closed loop, following similar arguments as in~Proposition~\ref{prop:impact_time}. In summary, the saturated system retains essential properties of~\eqref{eq:SpH}, but the convergence to~$\mathcal{S}$ is slower due to the saturation.
The jump across~$\mathcal{S}$ may still be a cause for chattering effects in practice, and next we remove the jump altogether.

\subsection{Linear Boundary Layer Approximation}
A common way to approximate a jump discontinuity is by using a linear interpolation of high gain in a boundary region; see~\cite[Figure~14.7]{khalil_nonlinear_2002}. In this vein, consider
\begin{equation*}
    \sigma^{-1}_\text{lin}(x) :=
    \begin{cases}
        1/\sigma(x), &\text{if $x \in  \mathcal{M}^\text{c}$} \\
        M^2 \sigma(x), & \text{if $x\in \mathcal{M}$}
    \end{cases},
\end{equation*}
which is linear in~$\sigma(x)$ inside~$\mathcal{M}$, with slope~$M^2$; see Fig.~\ref{fig:sigma_inverses}.
Note that $\sigma_\text{lin}^{-1}(x)=0$ on $\mathcal{S}$ and $\sigma_\text{lin}^{-1}(x)$ is continuous for all~$x$.
Consider the following regularized version of~\eqref{eq:SpH} that is continuous:
\begin{equation} \label{eq:sys_lin}
    \begin{aligned}
        \dot x & = [\bar J(x,t) - \sigma(x)R(x,t)]Qx + \sigma^{-1}_\text{lin}(x)\bar B(x)u \\
        y & = \bar B(x)^\top Qx.
    \end{aligned}
\end{equation}
In~$\mathcal{M}^\text{c}$, \eqref{eq:sys_lin} is identical to~\eqref{eq:SpH}. Inside~$\mathcal{M}$, we can differentiate $H$ and obtain
\begin{align*}
    \dot H(x) & = \sigma(x)x^\top Q \dot x =  -\sigma^2(x) x^\top QR(x,t)Q x + M^2\sigma^2(x)y^\top u \\
    & =  -2H(x) x^\top QR(x,t)Q x + 2M^2H(x)y^\top u.
\end{align*}
Because of the factor $2M^2H(x)$ in front of $y^\top u$, $H$ is not a storage function with respect to the supply rate~$y^\top u$. However, for $x\not\in\mathcal{S}$ and upon division by $2M^2 H(x)$, we can rewrite the equality as
\begin{align*}
   \frac{1}{2M^2} \frac{d}{dt} \ln H(x) = \underbrace{-\frac{1}{M^2}x^\top QR(x,t)Q x}_{\leq 0} + y^\top u.
\end{align*}
Similar to the previous systems, there is a storage function with respect to~$y^\top u$:
\begin{equation}\label{eq:Hlin}
    H_\text{lin}(x) = 
    \begin{cases}
        H(x) - \frac{1+\ln 2M^2}{2M^2}, &  \text{if $x\in \mathcal{M}^\text{c}$} \\
        \frac{1}{2M^2} \ln H(x), & \text{if $x\in  \mathcal{M}$}.
    \end{cases}
\end{equation}
The constant added to $H(x)$ in $\mathcal{M}^\text{c}$ is again introduced to achieve continuity of~$H_\text{lin}$ across~$\partial \mathcal{M}$. The gradient $\nabla H_\text{lin}(x)$ is continuous, also across $\partial \mathcal{M}$.
In contrast to the earlier storage functions, $H_\text{lin}(x)$ has \emph{no lower bound} since  $H_\text{lin}(x)\to -\infty$ as $x\to\mathcal{S}$. Hence, \eqref{eq:sys_lin} is not passive but rather \emph{cyclo-dissipative}~\eqref{eq:cyclodiss} with respect to the supply rate~$y^\top u$; see~\cite{willems_qualitative_1974,van_der_schaft_cyclo-dissipativity_2021}. This is also called cyclo-passivity~\cite{moylan_dissipative_2014}.
\begin{theorem}
Consider the regularized system~\eqref{eq:sys_lin}. Trajectories~$x(t)\not \in \mathcal{S}$ satisfy the differential dissipation inequality
\begin{equation*}
    \dot{H}_\text{lin}(x) = -d_\text{lin}(x,t) + y^\top u \leq y^\top u,
\end{equation*}
with storage function~$H_\text{lin}$ in~\eqref{eq:Hlin} and continuous dissipation rate
\begin{equation}
    d_\text{lin}(x,t) = 
    \begin{cases}
        \sigma^2(x) x^\top Q R(x,t) Q x, &  \text{if $x\in \mathcal{M}^\text{c}$} \\
        \frac{1}{M^2}x^\top QR(x,t)Q x, &  \text{if $x\in  \mathcal{M}$}
    \end{cases},
\end{equation}
establishing cyclo-passivity of~\eqref{eq:sys_lin}.
\end{theorem}

The system~\eqref{eq:sys_lin} is continuous, satisfies a passivity inequality, \emph{and} 
can provide an infinite amount of energy by means of its storage function with no lower bound. Hence, it retains essential properties of~\eqref{eq:SpH}, but not by means of a sliding mode.
We can interpret the cyclo-dissipativity of~\eqref{eq:sys_lin} as if the infinite available supply of~\eqref{eq:SpH} on~$\mathcal{S}$ has been ``distributed'' over the boundary layer~$\mathcal{M}$. In particular, as $M\to\infty$, $\mathcal{M}$ converges to $\mathcal{S}$ and the difference in practice is negligible. For instance, on the boundary $x\in\partial \mathcal{M}$, we have the storage $H_\text{lin}(x)=\frac{1}{2M^2} -  \frac{1+\ln 2M^2}{2M^2}\to 0$ as $M\to\infty$.

\section{Numerical Examples}\label{sec:examples}
In this section, we illustrate some results using simulations. We use the regularized system~\eqref{eq:sys_lin}, since~\eqref{eq:SpH} is hard to simulate directly due to the infinite singularity on~$\mathcal{S}$.

\begin{example}[Frequeny Tracking with Static Passive Load]
\label{ex:static_passive}
Consider a case closely resembling the scenario in Section~\ref{sec:stab_interconnect} in two dimensions with $x=(x_1,x_2)$. Choose
\begin{align*}
    \bar J = \begin{bmatrix}
    0 & -\omega_0 \\ \omega_0 & 0
\end{bmatrix}, \quad R=Q=I, \quad \bar B = \begin{bmatrix}
    1 \\ 0
\end{bmatrix}, \quad \omega_0= 2\pi \  \text{rad/s},
\end{align*}
and saturation~$M=3$.
Since $\bar B$ is not invertible and due to saturation, Proposition~\ref {prop:impact_time} cannot be applied directly. We simulate the system under two different (linear) interconnection conditions~\eqref{eq:static_load}: $K=1$ and $K=5$. The phase portraits are shown in Fig.~\ref{fig:phase_portrait_static} and the corresponding  first-element trajectories~$x_1(t)$, $t\geq 0$, in~Fig.~\ref{fig:time_domain_static}. The simulations show finite-time convergence to the desired frequency and amplitude ($\mathcal{S}=\{x\,:\, x^\top x=1\}$) under both interconnections and different initial states. Some transient chatter around~$\mathcal{S}$ can be observed, since the numerical solver does not hit the surface~$\mathcal{S}$ perfectly.

\begin{figure}
        \centering
        \includegraphics[width=0.9\columnwidth]{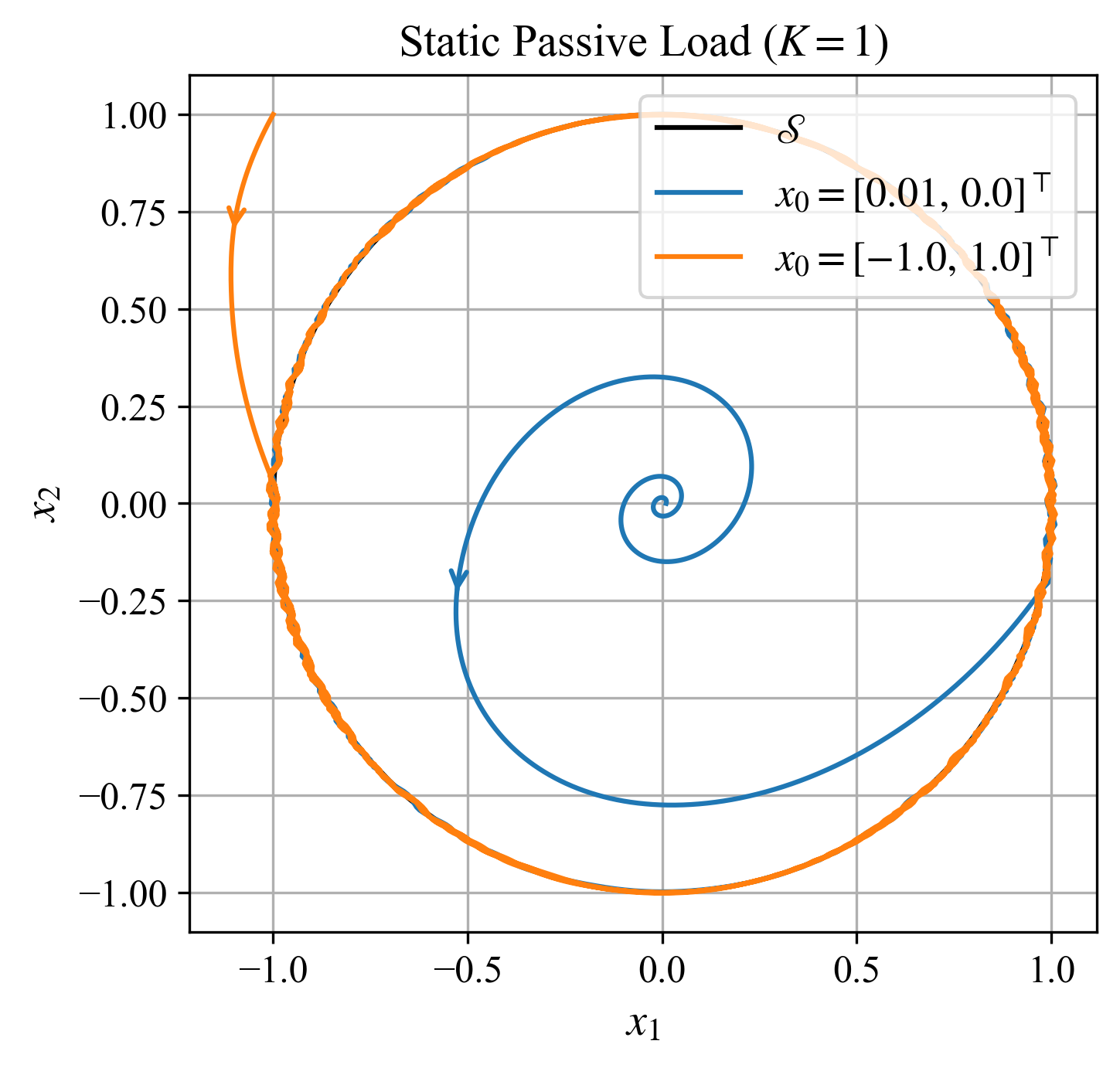}
        \includegraphics[width=0.9\columnwidth]{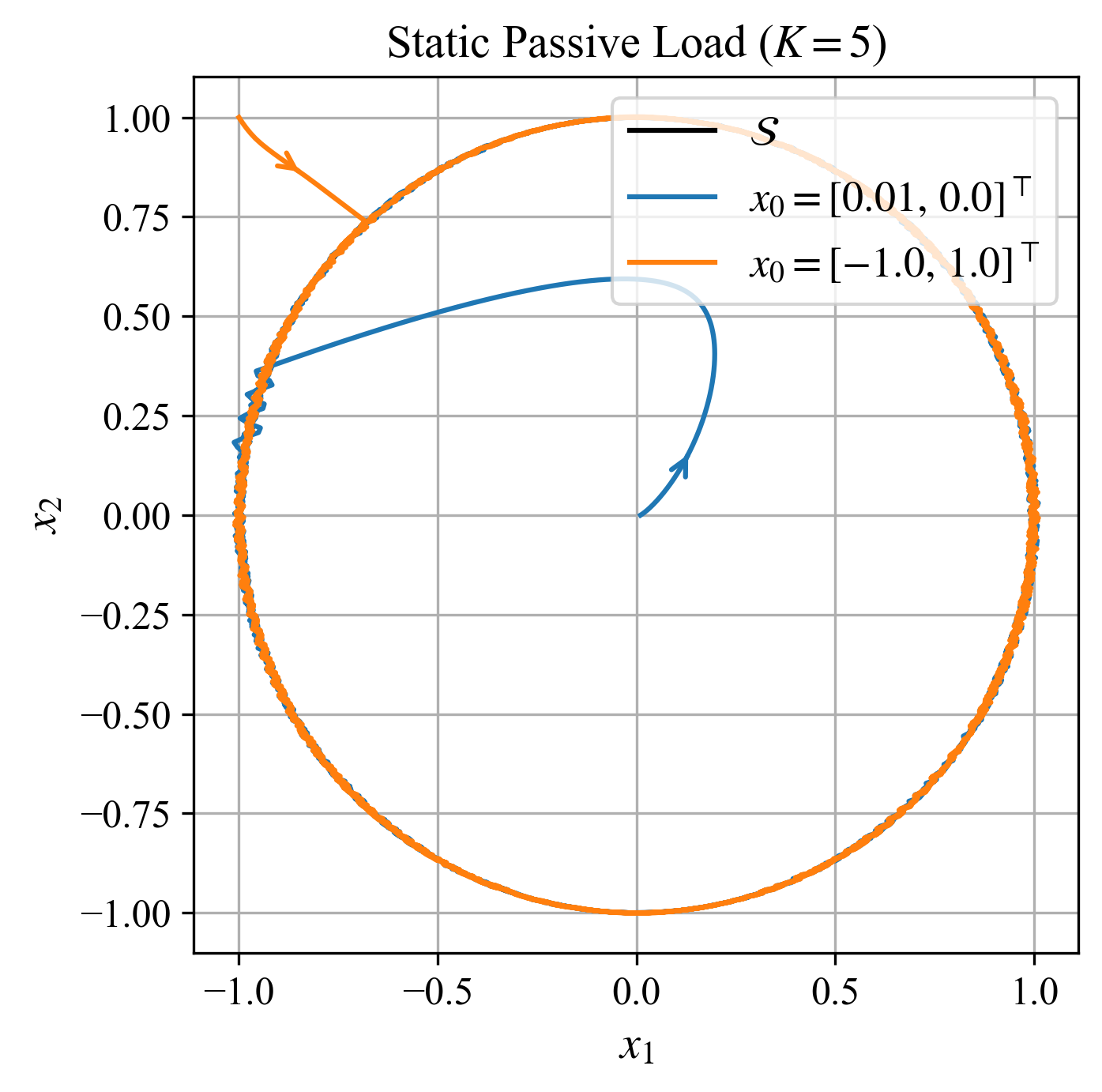}
        \caption{The phase portraits from Example~\ref{ex:static_passive}. The state converges quickly to~$\mathcal{S}$ for different initial states, under two different interconnections~$K$.}
        \label{fig:phase_portrait_static}
\end{figure}
\begin{figure}
        \centering
        \includegraphics[width=0.9\columnwidth]{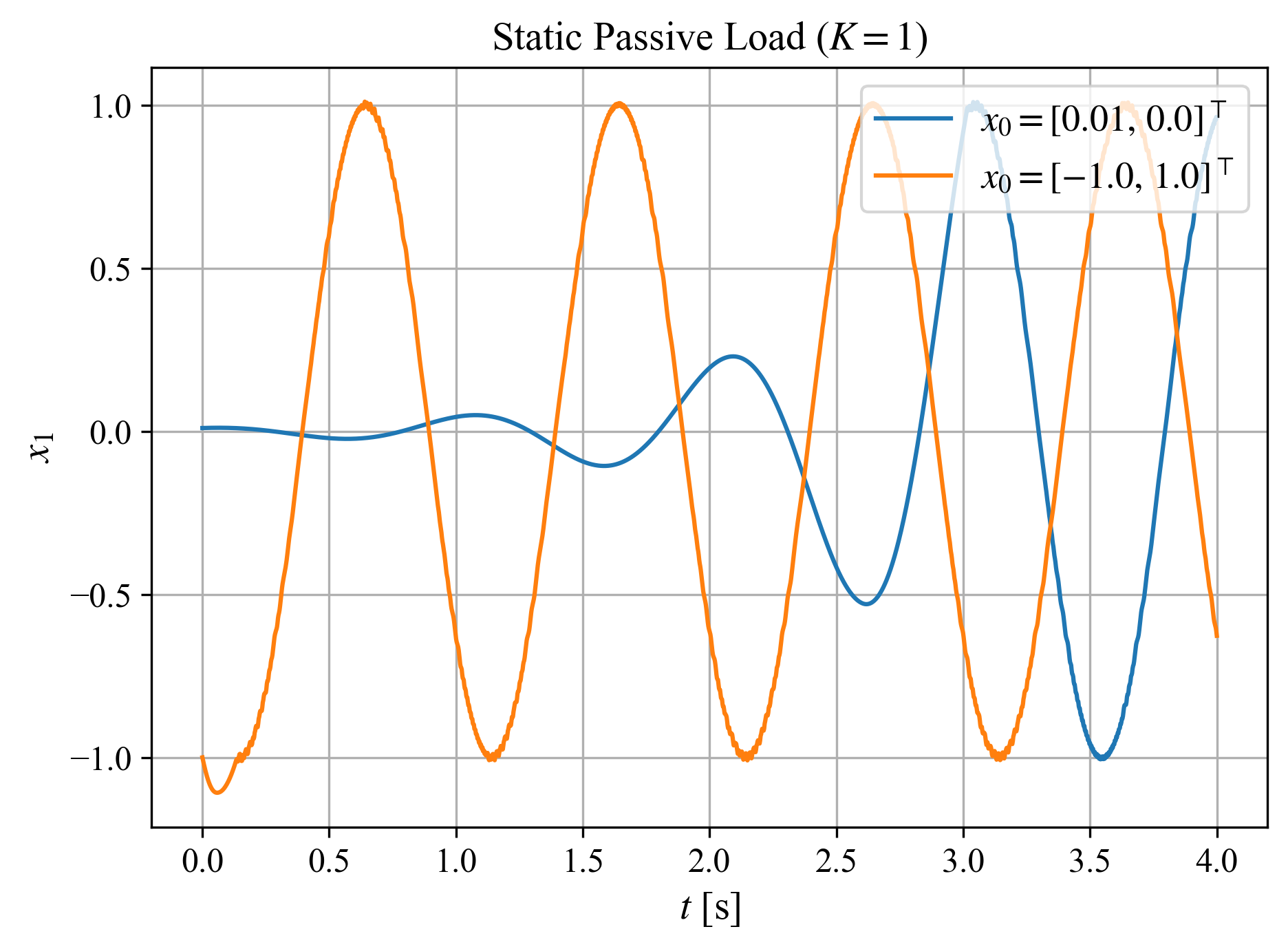}
        \includegraphics[width=0.9\columnwidth]{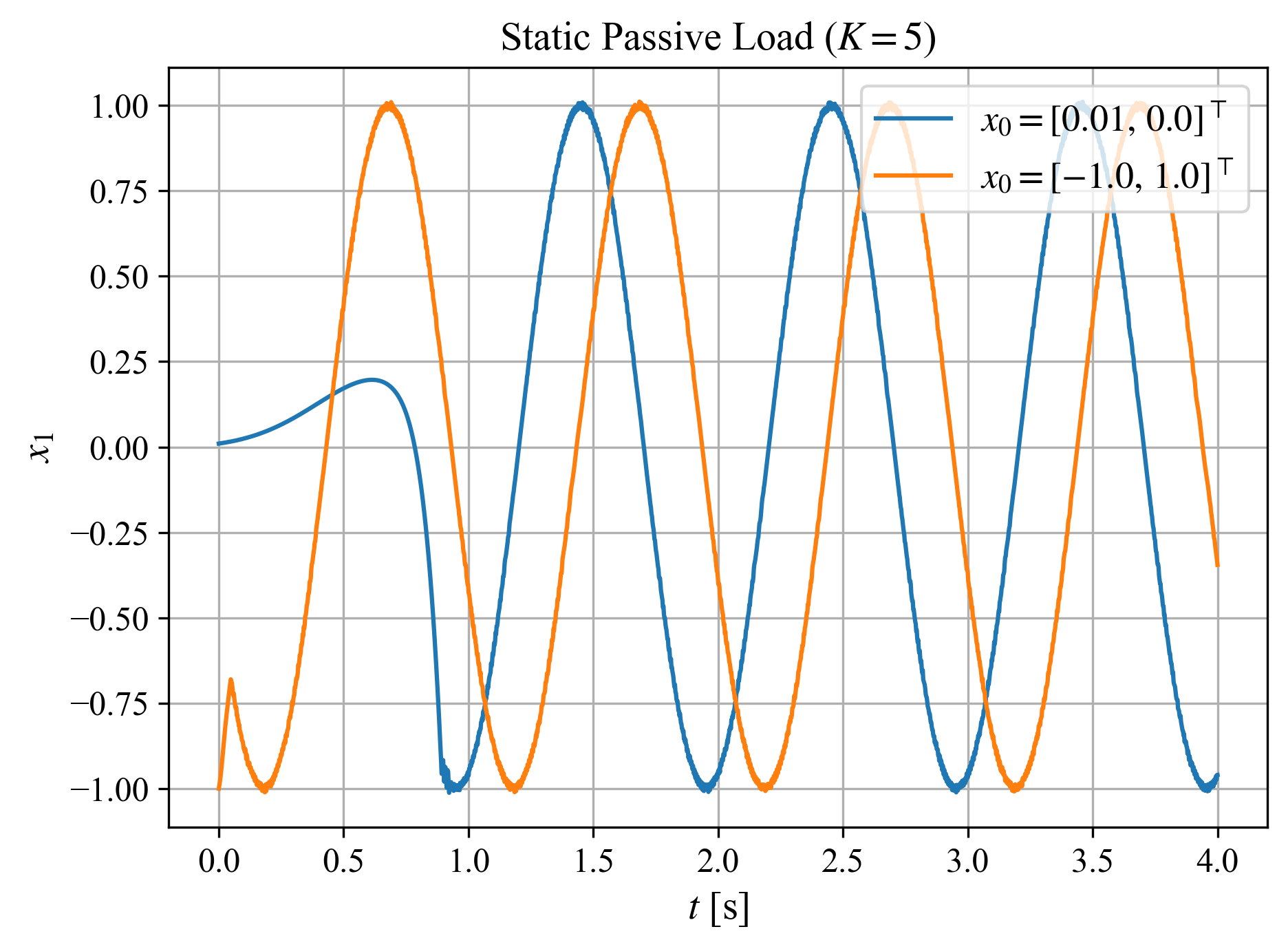}
        \caption{State trajectories~$x_1(t)$ from Example~\ref{ex:static_passive}. The state converges quickly to the desired frequency $\omega_0=2\pi$~\text{rad/s} of amplitude~$1$, under two different interconnections~$K$.}
        \label{fig:time_domain_static}
\end{figure}
\end{example}

\begin{example}[Phase Tracking with Dynamic Passive Load]
\label{ex:dynamic_passive}
We extend Example~\ref{ex:static_passive} and let the passive, interconnected system be \emph{dynamic} by using the positive real transfer function ($s$ is the Laplace variable)
\begin{equation*}
    K(s) = \frac{s+3}{s^2+4s+4},
\end{equation*}
and vary the dissipation in~\eqref{eq:sys_lin} by using~$R= r\cdot I_2$, $r>0$. Instead of tracking a frequency, suppose we would like the system state to converge to a point on~$\mathcal{S}$ with phase/argument $\phi_\text{ref}=3\pi/4 \ \text{rad}$. This can be achieved by adding ``phase PI control'' through the matrix $\bar J$ and an additional state variable~$x_\text{i}$:
\begin{align*}
    \bar J(x,x_i) & = \begin{bmatrix}
        0 & -\omega_0-k_\text{p}e(t)-k_\text{i}x_\text{i}(t) \\
        \omega_0+k_\text{p}e(t)+k_\text{i}x_\text{i}(t) & 0
    \end{bmatrix}\\
    e(t) & = \phi_\text{ref} - \arg(x_1(t)+jx_2(t)) \\
    \dot x_\text{i}(t)  & = e(t),
\end{align*}
where $j$ is the imaginary unit and the $x_\text{i}$-dynamics is augmented with~\eqref{eq:sys_lin}. We pick the tuning parameters $k_\text{p}=50$ and
$k_\text{i}=200$ next. 
Intuitively, the matrix $\bar J$ rotates the vector~$x$ until it stops at the desired angle, and the dissipation matrix $R\succ 0$ regulates convergence to~$\mathcal{S}$.

Simulations of the interconnected dynamics for three different values of $r$ are shown in Fig.~\ref{fig:phase_portrait_dynamic}. Interestingly, while the trajectories initially converge to~$\mathcal{S}$, the state appears to become ``charged'' and can pull out of~$\mathcal{S}$. Eventually, the state gets pulled back and converges to the desired target on~$\mathcal{S}$. By increasing the dissipation~$r$, the detour behavior can be controlled. The example shows that with dynamic passive loads, it can happen that the vector field points out of~$\mathcal{S}$, in contrast to the static case considered in Proposition~\ref{prop:impact_time}.
A rigorous stability analysis, in cases such as this one, remains an interesting direction for future work.
\end{example}
\begin{figure}
      \centering
      \includegraphics[width=0.9\columnwidth]{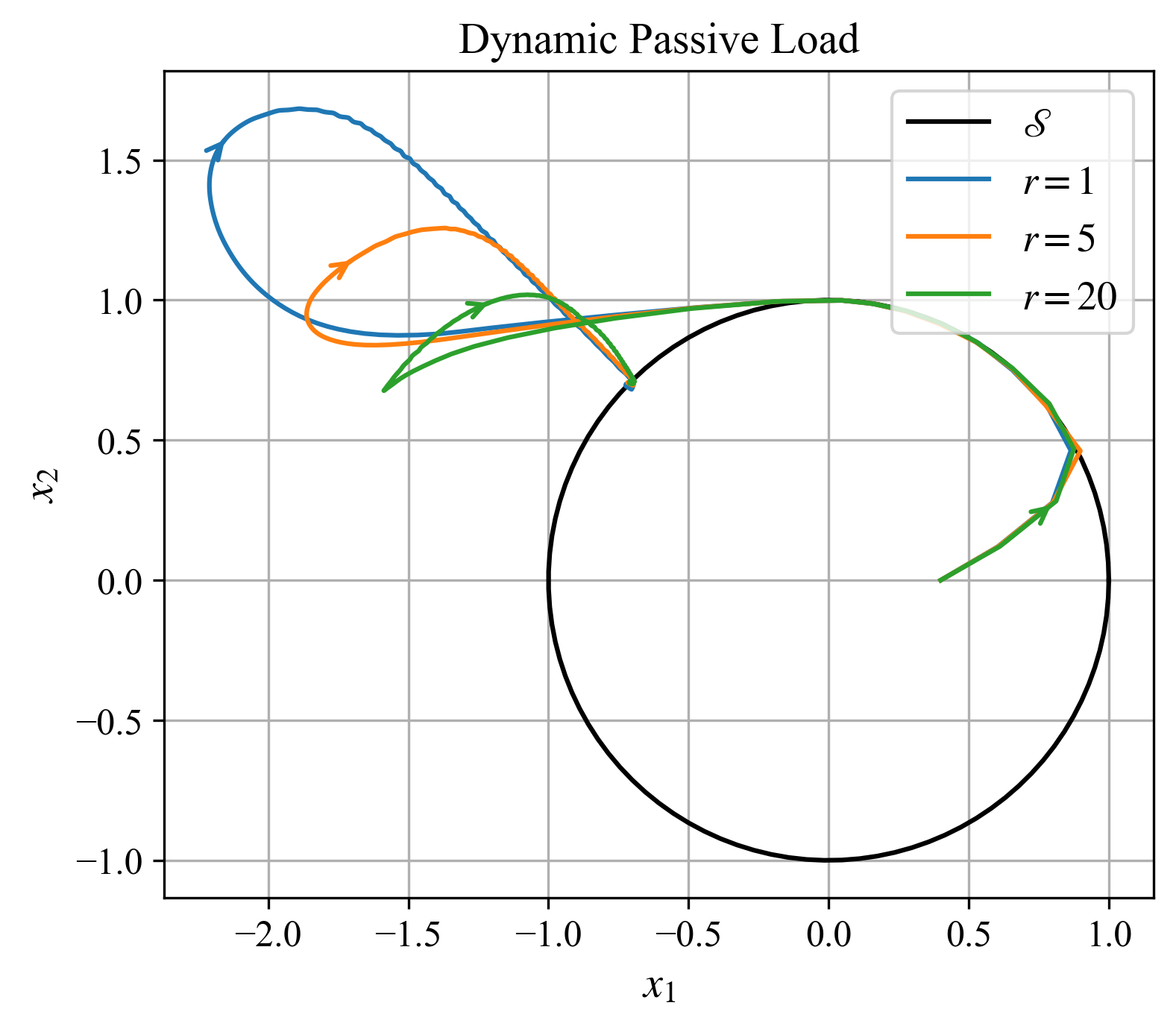}
      \caption{The phase portraits from Example~\ref{ex:dynamic_passive}. The state converges to the desired point on~$\mathcal{S}$ for different dissipations, but leaves the vicinity of~$\mathcal{S}$ briefly.}
      \label{fig:phase_portrait_dynamic}
\end{figure}


\section{Conclusions} \label{sec:conclusion}
In this paper, we introduce a class of singular port-Hamiltonian systems and analyze their properties. The systems are not passive; they can serve as infinite energy sources. Yet they satisfy local passivity inequalities that we used to prove convergence to a desired ellipsoidal surface upon interconnection with passive systems (loads). We have also investigated two related systems with regularized singularities to ease implementation. They both exhibit properties similar to those of the original class of systems, although one of them is in fact cyclo-passive. An interesting topic for future research is more general stability analysis under dynamic interconnections, as illustrated in Example~\ref{ex:dynamic_passive}. Furthermore, possible connections to the IDA--PBC framework~\cite{ortega_putting_2001,ortega_interconnection_2002} (see Remark~\ref{rem:IDA-PBC}), and dissipativity theory for switched and hybrid systems~\cite{zhao_dissipativity_2005,teel_asymptotic_2010} should also be investigated.


\section*{Usage of Generative AI}
ChatGPT~\cite{noauthor_chatgpt_nodate} was used to assist with language refinement.
Grammarly~\cite{noauthor_grammarly_nodate} was used for grammar and spelling checks. All content was reviewed and verified by the authors.


\printbibliography

\end{document}